\documentclass[conference, a4paper]{IEEEtran}
\usepackage[utf8]{inputenc} 
\usepackage[T1]{fontenc}
\IEEEoverridecommandlockouts
\usepackage{url}
\usepackage{cite}
\usepackage[cmex10]{amsmath}
\allowdisplaybreaks[4]
\usepackage{amsthm}
\usepackage{mleftright} 
\mleftright 
\usepackage{graphicx}
\usepackage{multirow}
\usepackage{booktabs}
\usepackage{algorithm, algorithmic}
\usepackage{color}
\usepackage{amssymb, amsfonts}
\usepackage{bm}
\usepackage{comment}
\usepackage{enumerate}
\usepackage{authblk}

\newtheorem{lemma}{Lemma}
\newtheorem{theorem}{Theorem}

\newtheorem{remark}{Remark}

\usepackage{geometry}
\title{Efficient Computation of Marton's Error Exponent via Constraint Decoupling
\thanks{The first three authors contributed equally to this work and $\dag$ marked the corresponding author. This work was supported by the National Natural Science Foundation of China (Grant Nos. 12271289 and 62231022).}}

\begin{document}

\newcommand{\thatis}{\textit{i.e.}}
\newcommand{\andsoon}{\textit{etc.}}

\newcommand{\mbbN}{\mathbb{N}} 
\newcommand{\mbbR}{\mathbb{R}} 
\newcommand{\mbbC}{\mathbb{C}} 
\newcommand{\mbbE}{\mathbb{E}}
\newcommand{\mbbP}{\mathbb{P}}

\newcommand{\mbfx}{\mathbf{x}} 
\newcommand{\mbfy}{\mathbf{y}} 
\newcommand{\mbfd}{\mathbf{d}} 
\newcommand{\mbfM}{\mathbf{M}}

\newcommand{\bda}{\boldsymbol{a}} 
\newcommand{\bdd}{\boldsymbol{d}} 
\newcommand{\bdf}{\boldsymbol{f}} 
\newcommand{\bdg}{\boldsymbol{g}} 
\newcommand{\bdp}{\boldsymbol{p}} 
\newcommand{\bdq}{\boldsymbol{q}} 
\newcommand{\bdr}{\boldsymbol{r}}
\newcommand{\bdu}{\boldsymbol{u}}  
\newcommand{\bdv}{\boldsymbol{v}} 
\newcommand{\bdw}{\boldsymbol{w}} 
\newcommand{\bdx}{\boldsymbol{x}} 
\newcommand{\bdy}{\boldsymbol{y}} 
\newcommand{\bdz}{\boldsymbol{z}}

\newcommand{\bdalpha}{\boldsymbol{\alpha}} 
\newcommand{\bdbeta}{\boldsymbol{\beta}} 
\newcommand{\bdphi}{\boldsymbol{\phi}} 
\newcommand{\bdpsi}{\boldsymbol{\psi}}

\newcommand{\mcX}{\mathcal{X}}
\newcommand{\mcP}{\mathcal{P}}
\newcommand{\mcY}{\mathcal{Y}}
\newcommand{\mcC}{\mathcal{C}}
\newcommand{\mcF}{\mathcal{F}}
\newcommand{\mcS}{\mathcal{S}}
\newcommand{\mcA}{\mathcal{A}}
\newcommand{\mcW}{\mathcal{W}}

\newcommand{\mrmd}{\mathrm{d}}
\newcommand{\mrme}{\mathrm{e}}

\newcommand{\bW}{\boldsymbol{W}}
\newcommand{\bA}{\boldsymbol{A}}
\newcommand{\bS}{\boldsymbol{S}}
\newcommand{\bX}{\boldsymbol{X}}
\newcommand{\ba}{\boldsymbol{\alpha}}
\newcommand{\bb}{\boldsymbol{\beta}}
\newcommand{\bg}{\boldsymbol{\gamma}}
\newcommand{\by}{\boldsymbol{y}}
\newcommand{\bo}{\boldsymbol{\omega}}
\newcommand{\bp}{\boldsymbol{p}}
\newcommand{\bM}{\boldsymbol{M}}
\newcommand{\bzero}{\boldsymbol{0}} 
\newcommand{\btwo}{\boldsymbol{2}} 
\newcommand{\bet}{\boldsymbol{\eta}} 

\author[1]{Jiachuan Ye}
\author[1]{Shitong Wu}
\author[1]{Lingyi Chen}
\author[2]{Wenyi Zhang}
\author[3]{Huihui Wu}
\author[1$\dag$]{Hao Wu}
\affil[1]{Department of Mathematical Sciences, Tsinghua University, Beijing 100084, China}
\affil[2]{Department of Electronic Engineering and Information Science, 
\authorcr University of Science and Technology of China, Hefei, Anhui 230027, China}
\affil[3]{Zhejiang Key Laboratory of Industrial Intelligence and Digital Twin, 
\authorcr Eastern Institute of Technology, Ningbo, Zhejiang 315200, P.R. China
\authorcr Email: hwu@tsinghua.edu.cn}

\maketitle

\begin{abstract}
The error exponent in lossy source coding characterizes the asymptotic decay rate of error probability with respect to blocklength. 
The Marton's error exponent provides the theoretically optimal bound on this rate. 
However, computation methods of the Marton's error exponent remain underdeveloped due to its formulation as a non-convex optimization problem with limited efficient solvers. 
While a recent grid search algorithm can compute its inverse function, it incurs prohibitive computational costs from two-dimensional brute-force parameter grid searches. 
This paper proposes a composite maximization approach that effectively handles both Marton's error exponent and its inverse function. 
Through a constraint decoupling technique, the resulting problem formulations admit efficient solvers driven by an alternating maximization algorithm. 
By fixing one parameter via a one-dimensional line search, the remaining subproblem becomes convex and can be efficiently solved by alternating variable updates, thereby significantly reducing search complexity. 
Therefore, the global convergence of the algorithm can be guaranteed. 
Numerical experiments for simple sources and the Ahlswede's counterexample, demonstrates the superior efficiency of our algorithm in contrast to existing methods. 
\end{abstract}


\section{Introduction}

The study of error exponents is a cornerstone in information theory, providing fundamental insights into the asymptotic performance limits of communication and compression systems. 
The earliest theoretical framework for error exponents was developed in the context of channel coding \cite{gallager1965coding,shannon1973dmc}, where Gallager and Shannon \textit{et al.} first quantified how the decoding error probability decays with increasing code length.
This research trajectory later extended to source coding, commencing with lossless compression \cite{jelinek1968probabilistic} before Blahut's advancement to lossy compression \cite{blahut1974hypothesis}, though the latter was subsequently shown to be suboptimal. 
%
%
Marton's derivation of the optimal error exponent for lossy source coding \cite{marton1974error}, which established the tighest asymptotic bound for optimal coding systems, now serves as a fundamental performance criterion and essential tool for practical code design evaluation and optimization.

Despite the theoretical significance of Marton's error exponent, practical computation remains challenging due to the non-convex optimization problem inherent in the definition of Marton's error exponent. 
This difficulty has hindered the development of efficient computational algorithms. 
%
%
Recent work \cite{jitsumatsu2023computation} proposed a grid search algorithm that establishes a two-parameter representation of the continuous inverse function of Marton's error exponent. 
However, the optimization problem requires brute-force two-dimensional search during computation, resulting severe computation inefficiency. 
Consequently, achieving high-precision solutions demands prohibitively expensive computational resources. 
To compute Marton's error exponent efficiently, we reformulate the problem using the dual form of the RD function, yielding an optimization model with a convex objective but non-convex constraints. 
By decoupling the RD subproblem, we obtain an equivalent model where two variables have convex constraints while a third remains non-convex. 
This allows a one-dimensional line search on the non-convex parameter, resulting in a convex subproblem involving two variables, where we have a closed-form solution for one variable solve for the other via the Blahut-Arimoto algorithm \cite{blahut1972computation}. 
Our algorithm could compute both Marton's error exponent and its inverse efficiently with global convergence. 
Experiments show an order-of-magnitude speedup over \cite{jitsumatsu2023computation} and successful handling of discontinuities in Ahlswede's counterexample \cite{ahlswede1990extremal}. 

\section{Marton's Error Exponent}
Consider a discrete memoryless source variable $X\in\mathcal{X}$ and a reconstruction variable $Y\in\mathcal{Y}$, where $|\mathcal{X}|=M$ and $|\mathcal{Y}|=N$ are finite. 
The set of probability distributions on $\mathcal{X}$ is denoted by $\mathcal{P(X)}$. 
Given a distortion measure $d(x,y)\geq 0$ between $\mathcal{X}$ and $\mathcal{Y}$, a distortion level $\Delta\geq 0$, and a marginal distribution $p_X\in\mathcal{P}(\mathcal{X})$, the rate-distortion (RD) function \cite{shannon1948mathematical} is defined by \cite{book_thomas_elements,book_berger_rd}
\fontsize{9pt}{12pt}{
\begin{equation}\label{rate-distortion}
R(\Delta,p_X)=\min\limits_{\substack{p_{Y|X}\in \mathcal{P(Y|X)}: \\ \mathbb{E}_{p_{X,Y}}[d(X,Y)]\leq\Delta}}I(X;Y),
\end{equation}
}where $\mathcal{P(Y|X)}$ is the set of conditional probability distributions on $\mathcal{Y}$ given $\mathcal{X}$, the expectation of $d(X,Y)$ is taken over the joint distribution $p_{X,Y}(x,y)=p_X(x)\cdot p_{Y|X}(y|x)$, and $I(X;Y)$ is the mutual information between $X$ and $Y$. 
We have $R(\Delta,p_X)=0$ if $\Delta\geq \Delta_{\max}:=\min\limits_{y\in\mathcal{Y}}\sum\limits_{x\in\mathcal{X}}p_X(x)d(x,y)$.
%


%
The RD theorem implies that $R(\Delta,p_X)$ is the infimum achievable rate for a given $\Delta$. 
When $R>R(\Delta,p_X)$, there exists a sequence of $k$-length block codes with rates converging to $R$ and the error probability $P^k(R,\Delta,p_X)=P\{x^k\in\mathcal{X}^k:d(x^k,\hat{x}^k)>\Delta\}$ converging to 0, where $\hat{x}^k$ is the reproduction of $x^k$, 
and $d(x^k,\hat{x}^k)=\dfrac{1}{k}\sum\limits_{i=1}^k d(x^k_i,\hat{x}^k_i)$ is the average distortion. 
A problem worth studying is the convergence rate of the error probability when the blocklength $k$ tends to infinity. 
Marton gave such an asymptotic behavior in \cite{marton1974error}. 
Fixing a distortion $\Delta\in[0,\Delta_{\max}]$ and a source distribution $q_X\in\mathcal{P(X)}$, Marton's error exponent is defined by
\fontsize{9pt}{12pt}{
\begin{equation}\label{marton-error-exponent}
E_{\rm M}(R,\Delta,q_X):= \min\limits_{\substack{p_X\in\mathcal{P(X)}: \\ R(\Delta,p_X)\geq R}}D_{\rm KL}(p_X\Vert q_X),
\end{equation}
}where $D_{\rm KL}(p_X\Vert q_X)$ denotes the Kullback-Leibler (KL) divergence between $p_X$ and $q_X$. The error exponent is proven to be optimal \cite{marton1974error}, \textit{i.e.} 
{\fontsize{9pt}{12pt}\selectfont
\begin{align*}E_{\rm M}(R-0,\Delta,q_X)\leq\varliminf\limits_{k\to\infty}\left[-\dfrac{1}{k}\log P^k(R,\Delta,q_X)\right] \\
\leq\varlimsup\limits_{k\to\infty}\left[-\dfrac{1}{k}\log P^k(R,\Delta,q_X)\right]\leq E_{\rm M}(R+0,\Delta,q_X).
\end{align*}
}
That is, the exponential convergence rate is tightly bounded by Marton's error exponent. 
On the other hand, \cite{jitsumatsu2023computation} also gives the inverse function of Marton's error exponent, which has the form
\begin{equation}\label{marton-inverse}
R_{\rm M}(E,\Delta,q_X):=\max\limits_{\substack{p_X\in\mathcal{P}(\mathcal{X}): \\ D_{\rm KL}(p_X\Vert q_X)\leq E}}R(\Delta,p_X).
\end{equation}

The optimality of Marton's error exponent underscores the theoretical importance of developing an efficient computational algorithm. 
However, the computation of this exponent remains challenging, primarily due to the structural properties therein. 
The twofold challenges lie in: (i) the nested optimization problem involving the RD function, and (ii) the non-concavity of $R(\Delta,p_X)$ in $p_X$ \cite{jitsumatsu2023computation}.
These challenges make it inherently difficult to formulate the Marton's error exponent as a tractable optimization model. 
%

%
%
To simplify the constraint in \eqref{marton-error-exponent}, the objective and constraints of the RD function must be decoupled. 
However, directly separating the objective from \eqref{rate-distortion} may introduce non-convex terms.
We therefore employ the dual form of the RD problem introduced later, which enables us to achieve decoupling and formulate a new set of constraints. 
This reformulation yields an analytically tractable yet equivalent model to \eqref{marton-error-exponent}, establishing the basis for our algorithm design.

\section{Composite Maximization Model}
To address the aforementioned difficulties in solving model \eqref{marton-error-exponent}, we employ an alternative definition of the RD function \cite[Theorem 2.5.3]{book_berger_rd} through the dual RD problem
\fontsize{9pt}{12pt}{
\begin{equation}\label{dual-rate-distortion}
R(\Delta,p_X)=\max\limits_{\zeta,\bda\geq 0}-\zeta\Delta+\sum\limits_i p_X(x_i)\log\dfrac{a_i}{p_X(x_i)},
\end{equation}
}
where $a_i$ are non-negative components satisfying the inequality constraints 
$\sum_i a_ie^{-\zeta d(x_i,y_j)}\leq 1,\forall j.$
The maximum is reached if and only if the equality holds. 
In the following, we denote symbols $p_i=p_X(x_i)$, $q_i=q_X(x_i)$, $d_{ij}=d(x_i,y_j)$. 
By doing so, we can reformulate the Marton's error exponent into a nested optimization model by substituting the definition of RD function in \eqref{dual-rate-distortion}:  
{\fontsize{9pt}{12pt}\selectfont
\begin{subequations}\label{er-nested}
\begin{align}
\max\limits_{\bdp}\quad&-D_{\rm KL}(\bdp\Vert\bdq), \label{er-nested-0} \\ 
\text{s.t.}\quad & \sum\limits_i a_i e^{-\zeta d_{ij}}=1,~\forall j;\quad\sum\limits_i p_i=1; \label{er-nested-2} \\
&\max\limits_{\zeta,\bda\geq 0}-\zeta\Delta+\sum\limits_i p_i\log\dfrac{a_i}{p_i}\geq R. \label{er-nested-1}
\end{align}
\end{subequations}
}
To compute $E_{\rm M}(R,\Delta,\bdq)$, we transform the saddle point calculation in \eqref{er-nested} into a maximization problem. 
The following Lemma 1 establishes the equivalence between this nested maximization model and a more tractable composite maximization model, while simultaneously providing the exact parametric form of Marton's error exponent.

\begin{lemma}
The nested maximization model \eqref{er-nested} of $E_{\rm M}(R,\Delta,\bdq)$ is equivalent to the following composite maximization model
\fontsize{9pt}{12pt}
\begin{subequations} \label{er-composite}
\begin{align} 
\max_{\bdp,\zeta,\bda}\quad & f_{\rm E}(\bdp)=-D_{\rm KL}(\bdp\Vert \bdq), \label{er-composite-0} \\
{\rm s.t.}\quad & \sum\limits_i a_i e^{-\zeta d_{ij}}=1,~\forall j;\quad\sum\limits_i p_i=1; \label{er-composite-2} \\ 
& -\zeta\Delta+\sum\limits_i p_i\log \dfrac{a_i}{p_i}\geq R. \label{er-composite-1} 
\end{align}
\end{subequations}
\end{lemma}
\begin{proof}
Suppose \eqref{er-composite} has an optimal triple $(\bdp^\star,\zeta^\star,\bda^\star)$. 
Then $\bdp^\star$ satisfies \eqref{er-nested-1} if we choose $\zeta^\star,\bda^\star$ here, and hence it is feasible for \eqref{er-nested}. On the other hand, suppose \eqref{er-nested} has an optimizer $\bdp^\star$, and then it satisfies \eqref{er-nested-1} with some $\zeta,\bda$, which indicates $(\bdp^\star,\zeta,\bda)$ is feasible for \eqref{er-composite}.
The optimal solutions of both models are feasible for each other, and hence they are equivalent. 
\end{proof}

%
Similarly, the inverse function of Marton's error exponent can be reformulated and equivalently converted into a composite optimization model, as stated in Lemma 2. 
\begin{lemma}
Computing $R_{\rm M}(E,\Delta,\bdq)$ is equivalent to solving the following composite maximization model
\begin{subequations} \label{re-composite}
\begin{align}
\max\limits_{\bdp,\zeta,\bda} \quad & f_{\rm R}(\bdp,\bda)=-\zeta\Delta+\sum\limits_i p_i\log\dfrac{a_i}{p_i}, \label{re-composite-0} \\
{\rm s.t.} \quad & \sum\limits_ia_ie^{-\zeta d_{ij}}=1,~\forall j;\quad\sum\limits_i p_i=1; \label{re-composite-2} \\
& D_{\rm KL}(\bdp\Vert\bdq)\leq E. \label{re-composite-1} 
\end{align}
\end{subequations}
\end{lemma}

\section{Alternating Maximization with Constraint Decoupling Algorithm}
In this section, we derive an algorithm to solve the Marton's error exponent \eqref{er-composite} and its inverse function \eqref{re-composite}. 
Noting that these two models are both jointly convex in $(\bdp,\bda)$, but non-convex with respect to $\zeta$, our proposed algorithm employs a two-layer optimization framework to address this challenge, consisting of the following main ingredients: 
\begin{itemize}
\item In the inner layer, optimize $(\bdp,\bda)$ for fixed $\zeta$, and output the corresponding optimal value. 
\item In the outer layer, conduct a one-dimensional line search, and select the $\zeta$ that yields the best overall performance across all searched points. 
\end{itemize}
Since the models \eqref{er-composite} and \eqref{re-composite} are derived through constraint decoupling, and the variables are updated via an alternating maximization scheme, we name our algorithm as Alternating Maximization with Constraint Decoupling (AM-CD) algorithm. 
The following two subsections present our algorithms for solving models \eqref{er-composite} and \eqref{re-composite}.

\subsection{AM-CD Algorithm for Marton's Error Exponent}
For fixed $\zeta$, we introduce multipliers $\eta\in\mathbb{R}$, $\bdalpha\in\mathbb{R}^M$, $\lambda\in\mathbb{R}_+$, and then the Lagrangian of \eqref{er-composite} is written as follows: 
{\fontsize{9pt}{12pt}\selectfont
\begin{equation*}
\begin{aligned}
&\mathcal{L}_{\rm E}(\bdp,\bda;\lambda,\eta,\bdalpha) = -\sum\limits_i p_i\log\frac{p_i}{q_i}-\eta\Big(\sum\limits_i p_i-1\Big) \\
&-\sum\limits_j\alpha_j\Big(\sum\limits_i a_ie^{-\zeta d_{ij}}-1\Big)+\lambda\Big(\sum\limits_i p_i\log\dfrac{a_i}{p_i}-\zeta\Delta-R\Big). 
\end{aligned}
\end{equation*}
}
We update the variables $\bdp,\bda$ in an alternating manner, thus updating $\bdp$ with fixed $\bda$, followed by updating $\bda$ with fixed $\bdp$. 
Derivations proceed as follows. 
First, we update the variable ${\bm p}$ by taking the derivative with respect to $p_{i}$, and the first-order condition yields
\begin{equation*}
\dfrac{\partial\mathcal{L}_{\rm E}}{\partial p_i}= -(\lambda+1)(\log p_i+1)+\log q_{i}+\eta+\lambda\log a_i=0.
\end{equation*}
Hence, for given $\bda$, the optimal $\bdp$ has the form
\begin{equation}\label{update_qi}
p_i=e^{\frac{\lambda\log a_{i}+\log q_{i}}{\lambda+1}}\Big/\Big(\sum_{i}e^{\frac{\lambda\log a_{i}+\log q_{i}}{\lambda+1}}\Big).
\end{equation}
Substituting it into \eqref{er-composite-1}, the dual variable $\lambda$ can be updated by finding the root of the following function
\begin{equation}\label{update_lambda}
\begin{aligned}
F_{\rm E}(\lambda)&\triangleq\frac{1}{\lambda+1} \Big(\sum_{i} e^{\frac{\lambda\log a_{i}+\log q_{i}}{\lambda+1}}\log \frac{a_{i}}{q_{i}}\Big) \Big/\Big(\sum\limits_i e^{\frac{\lambda\log a_{i}+\log q_{i}}{\lambda+1}}\Big) \\
&+\log\Big(\sum_{i}e^{\frac{\lambda\log a_{i}+\log q_{i}}{\lambda+1}}\Big)-\zeta\Delta-R.
\end{aligned}
\end{equation}
As derived in Appendix A, straightforward calculations yield $F_{\rm E}(0)<0$ and $F_{\rm E}'(\lambda)>0$. 
The monotonicity of $F_{\rm E}(\lambda)$ thus guarantees the existence of a non-negative root, which can be efficiently computed via Newton's method. 
Next, when $\bdp$ and $\lambda$ are fixed, the Lagrangian only depends on $\bda$ with its dual variable $\bdalpha$, \textit{i.e.}
{\fontsize{9pt}{12pt}\selectfont
\begin{align*}
\widetilde{\mathcal{L}_{\rm E}}&=\lambda\Big(\sum\limits_i p_i\log\dfrac{a_i}{p_i}-\zeta\Delta-R\Big)-\sum\limits_j\alpha_j\Big(\sum\limits_i a_ie^{-\zeta d_{ij}}-1\Big).
\end{align*}
}
Notably, optimizing the Lagrangian is exactly equivalent to solving the dual RD problem. 
Therefore, we adopt the following approach: 
we introduce auxiliary variables $\bdw$ and $\bdr$ where $w_{ij}=P_{Y|X}(y_j|x_i)$, $r_j=P_Y(y_j)$, and turn to solve the original RD problem with the Blahut-Arimoto (BA) algorithm \cite{blahut1972computation}, where $\bdw$ and $\bdr$ are updated in an alternate fashion as follows: 
{\fontsize{9pt}{12pt}\selectfont
\begin{align*}
w_{ij}=e^{-\zeta d_{ij}}r_j\Big/\Big(\sum\limits_j e^{-\zeta d_{ij}}r_j\Big),\quad r_j=\sum\limits_i p_iw_{ij}.
\end{align*}
}
Once the optimal $\bdw$ and $\bdr$ are obtained, we can update $\bda$ as follows\cite{book_berger_rd}, thereby solving the dual RD problem:  
{\fontsize{9pt}{12pt}\selectfont
\begin{equation}\label{update-a}
a_i=p_i\Big/\Big(\sum\limits_j e^{-\zeta d_{ij}}r_j\Big).
\end{equation}
}

\begin{algorithm}[htbp]
\begin{algorithmic}[0]
\caption{AM-CD Algorithm for Marton's Error Exponent \eqref{marton-error-exponent}}
\label{algorithm1}
\REQUIRE{Distortion metric $\bdd$, source distribution $\bdq$, rate threshold $R$, distortion threshold $\Delta$; mesh size $\delta$}
\FOR{$k=1:K$}
    \STATE Set $\zeta_k=k\delta$, $r_j=1/N$
    \FOR{${\rm outer\_iter}=1:{\rm max\_outer\_iter}$}
        \STATE Update $\lambda$ by solving the root of $F_{\rm E}(\lambda)$
        \STATE $p_i\leftarrow e^{\frac{\lambda\log a_{i}+\log q_{i}}{\lambda+1}}\Big/\Big(\sum\limits_i e^{\frac{\lambda\log a_{i}+\log q_{i}}{\lambda+1}}\Big)$
        \FOR{${\rm inner\_iter}=1:{\rm max\_inner\_iter}$}
            \STATE $w_{ij}\leftarrow e^{-\zeta_k d_{ij}}r_j\Big/\Big(\sum\limits_j e^{-\zeta_k d_{ij}}r_j\Big)$ 
            \STATE $r_j\leftarrow \sum\limits_i w_{ij}p_{i}$
        \ENDFOR 
        \STATE $a_i\leftarrow p_i\Big/\Big(\sum\limits_j e^{-\zeta_k d_{ij}}r_j\Big)$ 
    \ENDFOR
    \STATE $E_k=-D_{\rm KL}(\bdp\Vert\bdq)$
\ENDFOR
\RETURN $E=\max\limits_{1\leq k\leq K}E_k$
\end{algorithmic}
\end{algorithm}
\begin{remark}
The first term of the Lagrangian $\widetilde{\mathcal{L}_{\rm E}}$ matches the objective of the dual RD problem \eqref{dual-rate-distortion}, while the second term encodes the constraints of $\bda$. This is precisely the Lagrangian of \eqref{dual-rate-distortion} up to a factor $\lambda$. Thus, optimizing it solves the dual RD problem.
Moreover, the optimal $\bda$ can be analytically derived from the optimizer $(\bdw,\bdr)$ of the primal RD problem, and thus we solve it before updating $\bda$. 
\end{remark}

Building upon the aforementioned algorithm design and incorporating a one-dimensional line search over $\zeta$, we present the complete optimization procedure in Algorithm~\ref{algorithm1}.
We now analyze the convergence of the proposed AM-CD algorithm, which is formally stated in the following theorem. 
\begin{theorem}
The limit point $(\bdp^\star,\bda^\star)$ of the sequence $\{(\bdp^{(n)},\bda^{(n)})\}$ is the optimizer of \eqref{er-composite} with fixed $\zeta$, where $(\bdp^{(n)},\bda^{(n)})$ is the point obtained in the $n$-th iteration.
\end{theorem}

\textit{Sketch of Proof: } The proof proceeds via three key steps: direct computation gives the descent estimation of the objective; the sequence $\{(\bdp^{(n)},\bda^{(n)})\}$ has a limit point $(\bdp^\star,\bda^\star)$ ensured by the Pinsker's inequality; $(\bdp^\star,\bda^\star)$ is proven to satisfy the Karush-Kuhn-Tucker (KKT) conditions. 
Details of derivation are presented in Appendix B. 

\subsection{The AM-CD Algorithm for the Inverse Function of Marton's Error Exponent}
For fixed $\zeta$, 
we introduce multipliers $\eta\in\mathbb{R}$, $\bdalpha\in\mathbb{R}^M$, $\xi\in\mathbb{R}_+$, and then the Lagrangian of \eqref{re-composite} is written as follows: 
\begin{equation*}
\begin{aligned}
\mathcal{L}_{\rm R}&(\bdp,\bda;\xi,\eta,\bdalpha)=-\zeta\Delta+\sum\limits_i p_i\log\dfrac{a_i}{p_i}-\eta\Big(\sum\limits_i p_i-1\Big) \\
&+\sum\limits_j\alpha_j\Big(\sum\limits_i a_ie^{-\zeta d_{ij}}-1\Big)-\xi\Big(\sum\limits_i p_i\log\dfrac{p_i}{q_i}-E\Big).
\end{aligned}
\end{equation*}
Similarly, we derive the AM-CD algorithm to update $\bdp$ and $\bda$ in an alternate fashion. 
First, we update $\bdp$ by taking the partial derivatives with respect to each $p_i$, and then the first-order condition yields
{\fontsize{9pt}{12pt}\selectfont
\begin{align*}
\dfrac{\partial\mathcal{L}}{\partial p_i}=-1-\log p_i+\log a_i-\eta-\xi\left(1+\log p_i-\log q_i\right)=0.
\end{align*}}
Hence, for given ${\bm a}$, the optimal $\bdp$ has the form 
\begin{equation}\label{re-update-p}
p_i=e^{\frac{\log a_i+\xi\log q_i}{1+\xi}}\Big/\Big(\sum\limits_i e^{\frac{\log a_i+\xi\log q_i}{1+\xi}}\Big).
\end{equation}
Substituting it into \eqref{re-composite-1}, we can update $\xi$ by finding the root of 
\begin{equation}\label{update-xi}
\begin{aligned}
F_{\rm R}(\xi)&\triangleq\dfrac{1}{1+\xi}\Big(\sum\limits_i e^{\frac{\log a_i+\xi\log q_i}{1+\xi}}\log\dfrac{a_i}{q_i}\Big)\Big/\Big(\sum\limits_i e^{\frac{\log a_i+\xi\log q_i}{1+\xi}}\Big) \\
&-\log\Big(\sum\limits_i e^{\frac{\log a_i+\xi\log q_i}{1+\xi}}\Big)-E.
\end{aligned}
\end{equation}
We have $F_{\rm R}'(\xi)<0$ as derived in Appendix A, and hence $F_{\rm R}(\xi)$ is monotonic. 
If $F_R(0)\geq 0$, then $F_R(\xi)$ is guaranteed to have a non-negative root, which can be efficiently computed via Newton's method. 
In the case of $F_R(0)<0$, the constraint \eqref{re-composite-1} corresponding to $\xi$ has already been satisfied, and thus we can simply set $\xi=0$ without root finding process. 
Next, when $\bdp$ and $\lambda$ are fixed, an analogous analysis to the previous subsection demonstrates that we can apply the BA algorithm to solve the primal RD problem, with $\bda$ updated thereafter through the update rule \eqref{update-a}. 

Based upon the algorithmic design, we conduct a one-dimensional line search over $\zeta$ to finish the optimization procedure, as presented in Algorithm~\ref{algorithm2}. 

Following the approach in the previous section, we state the convergence behavior of our algorithm in the following theorem, with the proof presented in Appendix B. 
\begin{theorem}
The limit point $(\bdp^\star,\bda^\star)$ of the sequence $\{(\bdp^{(n)},\bda^{(n)})\}$ is the optimizer of \eqref{re-composite} with fixed $\zeta$, where $(\bdp^{(n)},\bda^{(n)})$ is the point obtained in the $n$-th iteration. 
\end{theorem}

\section{Experiments and Discussions}
The performance of the proposed algorithm is validated through some experiments in this section. 
These experiments have been conducted by Matlab 2023b on a laptop with 16G RAM and one Intel(R) Core(TM) i7-12700H CPU @ 2.30GHz. 
\subsection{Performance on Classical Distributions}
We consider the Gaussian source with squared error distortion and the Laplacian source with absolute error distortion in this subsection. 
We truncate the continuous distribution with an interval $[-L,L]$ and discretize it with a uniform grid $x_i=-L+(i-1/2)\delta$, $\delta=2L/M$, $1\leq i\leq M$, where $N=M=100$ in our cases. 
We take $\Delta=0.4$ for both cases. 
We set the Gaussian source with mean $\mu=0$, variance $\sigma=1$, and the Laplacian source with scaling parameter $b=1$. 
\begin{algorithm}[hbtp]
\begin{algorithmic}[0]
\caption{AM-CD Algorithm for the Inverse Function of Marton's Error Exponent\eqref{marton-inverse}}
\label{algorithm2}
\REQUIRE{Distortion metric $\bdd$, source distribution $\bdq$, error threshold $E$, distortion threshold $\Delta$; mesh size $\delta$}
\FOR{$k=1:K$}
    \STATE Set $\zeta_k=k\delta$, $r_j=1/N$
    \FOR{${\rm outer\_iter}=1:{\rm max\_outer\_iter}$}
        \STATE Update $\xi$ by solving the root of $F_{\rm R}(\xi)$
        \STATE $p_i\leftarrow e^{\frac{\log a_i+\xi\log q_i}{1+\xi}}\Bigg/\left(\sum\limits_i e^{\frac{\log a_i+\xi\log q_i}{1+\xi}}\right)$
        \FOR{${\rm inner\_iter}=1:{\rm max\_inner\_iter}$} 
            \STATE $w_{ij}\leftarrow e^{-\zeta_k d_{ij}}r_{j}\Big/\Big(\sum\limits_j e^{-\zeta_k d_{ij}}r_j\Big)$ 
            \STATE $r_j\leftarrow \sum\limits_i w_{ij}p_i$
        \ENDFOR 
        \STATE $a_i\leftarrow p_{i}\Big/\Big(\sum\limits_j e^{-\zeta_k d_{ij}}r_{j}\Big)$ 
    \ENDFOR
    \STATE $R_k=-\zeta_k\Delta+\sum\limits_i p_i\log\dfrac{a_i}{p_i}$
\ENDFOR
\RETURN $R=\max\limits_{1\leq k\leq K}R_k$
\end{algorithmic}
\end{algorithm}

We first compute the inverse function of Marton's error exponent, and compare the performance of our algorithm against that of grid search (Grid) algorithm \cite{jitsumatsu2023computation}. 
The results are reported in the Table I, where each result is obtained by repeating the experiment 100 times.
As shown in Table~\ref{table:re-comparison}, our AM-CD algorithm offers tens of times of speedup over the grid search algorithm. 
The significant efficiency gain stems from replacing the two-dimensional grid search with a hybrid one-dimensional line search and Newton's method, which also improves accuracy: Newton's search precisely computes the multiplier $\xi$ and outputs the rate $R$, while the grid search algorithm's accuracy depends on quadratically expensive grid refinement. 
%


%
\begin{table}[ht]
    \centering
    \caption{Comparison of efficiency between the AM-CD algorithm and the grid search algorithm on computing the inverse function of Marton's error exponent. }
    \begin{center}
    
    \begin{tabular}{|c|c|c|c|c|}
    \hline
    \multirow{2}{*}{Source} & \multirow{2}{*}{$(E,R)$} & \multicolumn{2}{c|}{Time(s)} & \multicolumn{1}{c|}{Speedup} \\ \cline{3-4}
    \multicolumn{1}{|c|}{} & \multicolumn{1}{c|}{} & AM-CD & Grid & \multicolumn{1}{c|}{Ratio} \\ \cline{1-5}
    \multirow{3}{*}{Gaussian} & $(0.10,0.7440)$ & 2.2174 & 82.4059 & 37.1836 \\
    \multicolumn{1}{|c|}{} & $(0.15,0.8007)$ & 1.9647 & 83.8407 & 42.6735 \\
    \multicolumn{1}{|c|}{} & $(0.20,0.8466)$ & 1.9529 & 83.3987 & 42.7051 \\ \cline{1-5}
    \multirow{3}{*}{Laplacian} & $(0.20,1.3433)$ & 2.6715 & 82.1908 & 30.7658 \\
    \multicolumn{1}{|c|}{} & $(0.25,1.3836)$ & 2.4574 & 82.4103 & 33.5356 \\
    \multicolumn{1}{|c|}{} & $(0.30,1.4170)$ & 2.3334 & 82.6220 & 35.4084 \\ \hline
    \end{tabular}
    \label{table:re-comparison}
    \end{center}
    \footnotesize{Note: The AM-CD algorithm stops if the decrease of $R$ between adjacent steps is less than $10^{-5}$. 
    For the AM-CD algorithm, we search $\zeta$ over a uniform grid of $[0,5]$ of size $100$. 
    For the grid search algorithm\cite{jitsumatsu2023computation}, we search over a uniform grid of $[0,5]\times[0,5]$ of size $100\times 100$. }
\end{table}
%


%
Different from the grid search algorithm, our approach can also compute Marton's error exponent directly. 
The results are reported in Table \ref{table:er}, where the same distributions are taken. 
Once again, our algorithm achieves fast computation of Marton's error exponent, with execution time comparable to that of computing its inverse function in the previous experiment.


\begin{table}[htbp]
    \centering
    \caption{Efficiency of the AM-CD algorithm on computing Marton's error exponent. }
    \begin{center}
    \begin{tabular}{|c|c|c|}
    \hline
    Source & $(R,E)$ & Time(s) \\ \cline{1-3}
    \multirow{3}{*}{Gaussian} & $(0.6,0.0222)$ & 0.9105 \\
    \multicolumn{1}{|c|}{} & $(0.7,0.0693)$ & 1.8439 \\
    \multicolumn{1}{|c|}{} & $(0.8,0.1492)$ & 2.4565 \\ \cline{1-3}
    \multirow{3}{*}{Laplacian} & $(1.1,0.0359)$ & 2.7047 \\
    \multicolumn{1}{|c|}{} & $(1.2,0.0816)$ & 3.0648 \\
    \multicolumn{1}{|c|}{} & $(1.3,0.1554)$ & 3.8876 \\ \hline
    \end{tabular}
    \end{center}
    \label{table:er}
    \footnotesize{Note: The algorithm stops if the decrease of $E$ between adjacent steps is less than $10^{-5}$. 
    We search $\zeta$ over a uniform grid of $[0,5]$ of size $100$. }
\end{table}


%
We also plot Marton’s error exponent and its inverse function computed by our AM-CD algorithm. As shown in Figures \ref{fig:er} and \ref{fig:re}, the algorithm consistently produces smooth, complete curves for both functions, demonstrating the method’s effectiveness. 
\begin{figure}[htbp]
    \centering
    \includegraphics[width=0.9\linewidth]{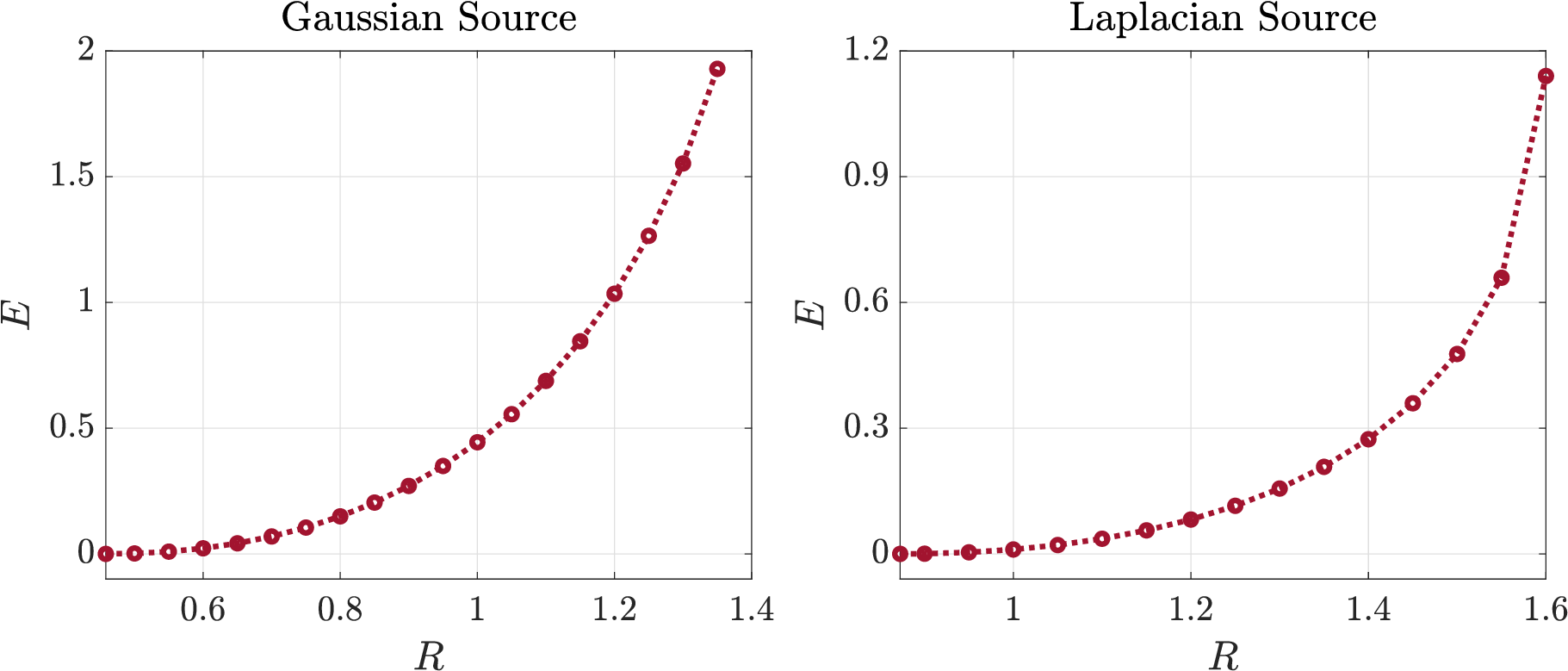}
    \caption{The curve of Marton's error exponent. }
    \label{fig:er}
\end{figure}
\begin{figure}[htbp]
    \centering
    \includegraphics[width=0.9\linewidth]{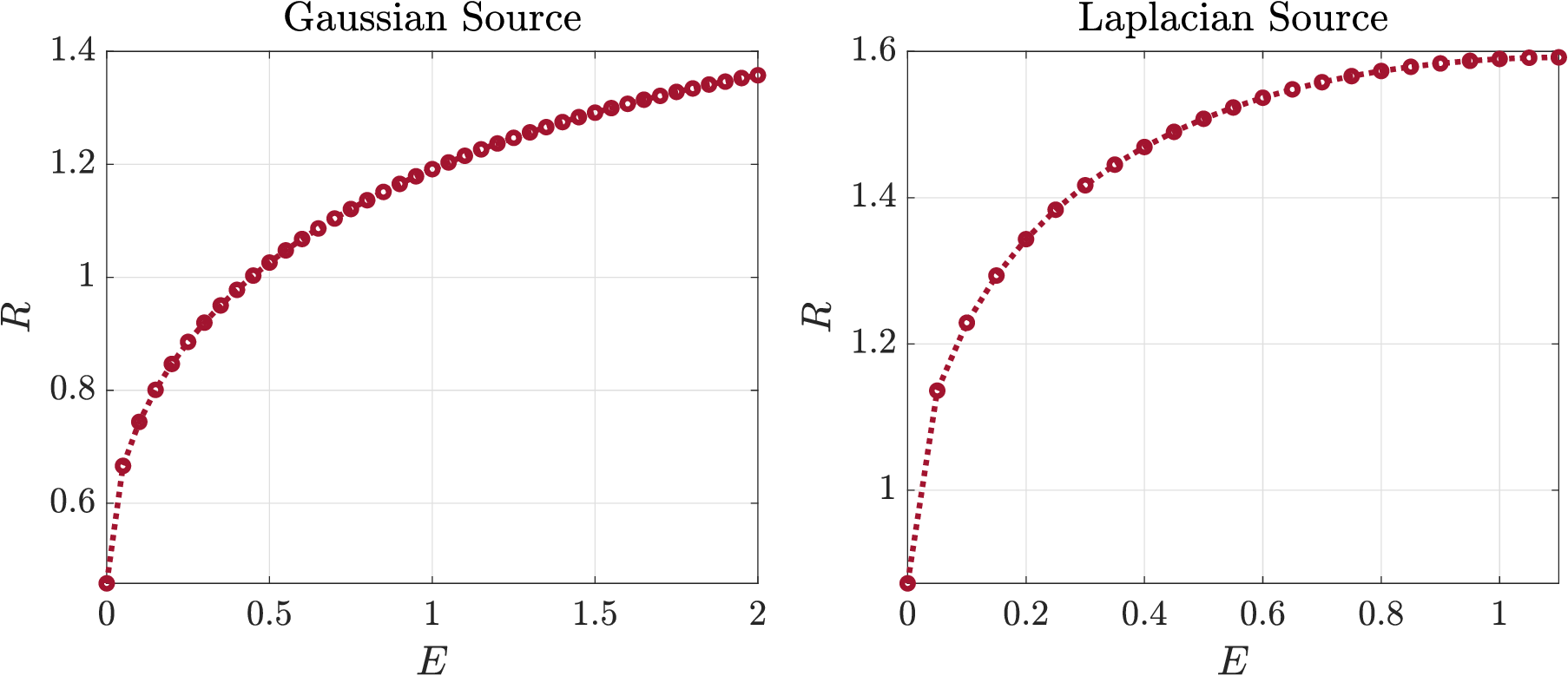}
    \caption{The curve of the inverse function of Marton's error exponent. }
    \label{fig:re}
\end{figure}

\subsection{Ahlswede's Counterexample}
This subsection discusses Ahlswede's counterexample where Marton's error exponent is a discontinuous function of rate $R$.
Let $\mathcal{Y}=\mathcal{X}$ with $\mathcal{X}$ partitioned into $\mathcal{X}_A$ and $\mathcal{X}_B$. 
We follow the definition in \cite[Section IV]{jitsumatsu2023computation}, where $d(x,y)$ is the distortion metric between $\mathcal{X}$ and $\mathcal{Y}$, $Q_A$ and $Q_B$ are uniform distributions on $\mathcal{X}_A$ and $\mathcal{X}_B$. 
For $\xi\in[0,1]$, define $Q_\xi=\xi Q_A+(1-\xi)Q_B$. 
Then the theoretical curve is derived from \cite[Theorem 3]{jitsumatsu2023computation} as follows:   
\[E_{\rm M}(R,\Delta,Q_\xi)
=\min\limits_{\substack{\lambda\in[0,1] \\ R(\Delta|Q_\lambda)\geq R}}D_2(\lambda\Vert\xi),\]
where $D_2(p\Vert q)=p\log(p/q)+(1-p)\log((1-p)/(1-q))$. 

We assume $|\mathcal{X}_B|=|\mathcal{X}_A|^3$, $|\mathcal{X}_A|=8$, $\Delta=0.254$, and $a=0.340$, where $a$ is a parameter defining $d(x,y)$. 
We then apply our algorithm to compute the inverse function of Marton's error exponent for the parameters defined by $q_X=Q_\xi$ with $\xi=0.01$. 
The theoretical curve and the reproduced curve are shown in Figure \ref{fig:ehlswede}. 
\begin{figure}[htbp]
\centering
\includegraphics[width=0.45\textwidth]{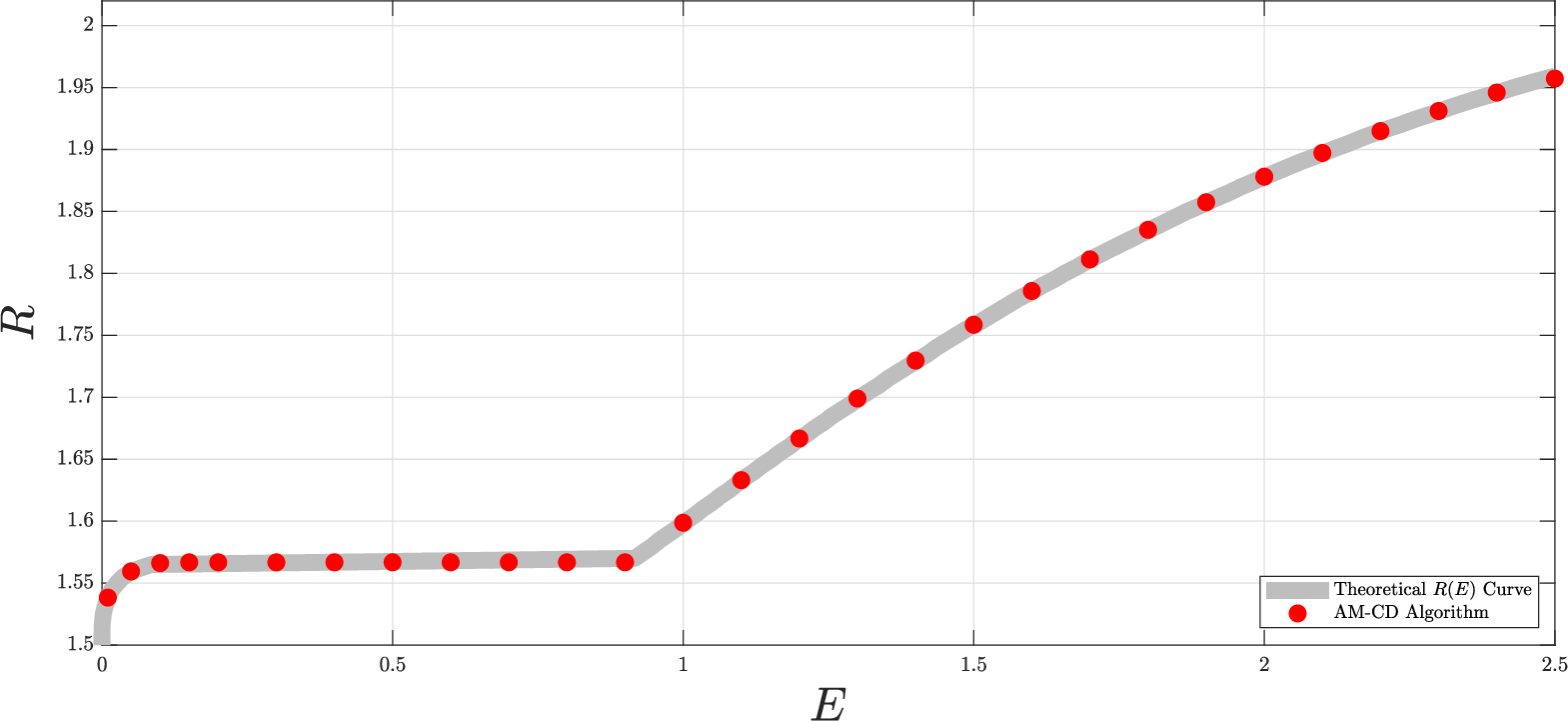}
\caption{The theoretical and reproduced curves of Marton's error exponent in Ahlswede's counterexample. }
\label{fig:ehlswede}
\end{figure}
At $R=1.5669$, the error exponent jumps from $E_L=0.1255$ to $E_R=0.9147$, demonstrating the discontinuity of Marton's error exponent at this point. 
Our algorithm successfully recovers the segment which is not continuous with respect to $R$. 
Moreover, our algorithm fits the theoretical curve well in continuous segments. 

\section{Conclusion}
This work presents an efficient computational approach for Marton’s error exponent and its inverse. 
By reformulating Marton’s error exponent as a nested maximization problem, we apply constraint decoupling to derive an equivalent composite maximization model. 
This novel formulation enables the development of the Alternating Maximization with Constraint Decoupling (AM-CD) algorithm, which replaces the computationally expensive two-dimensional grid search algorithm with a more efficient combination of one-dimensional line search and convex optimization. 
Numerical experiments validate the effectiveness of our method, and its robust performance on Ahlswede’s counterexample further confirms its reliability. 
We believe that this approach holds promise for broad applicability and anticipate its extension to computing other error exponents and strong converse exponents, such as those of Dueck and K\"orner \cite{dueckkorner1979reliability}, of Csiszar and K\"orner \cite{csiszar2011information}, and of Tridenski and Zamir \cite{tridenski2017exponential}.

\newpage
\bibliographystyle{IEEEtran}
\bibliography{ref}

\newpage
\indent
\newpage

\begin{appendices}

\section{Properties of Functions $F_{\rm E}(\lambda)$ and $F_{\rm R}(\xi)$}
In Section IV, we propose the root-finding of monotonic functions $F_{\rm E}(\lambda)$ and $F_{\rm R}(\xi)$ via Newton's method, which facilitates the updates of dual variables $\lambda$ and $\xi$. 
In this section, we provide rigorous analysis of these two functions. 
\subsubsection{The Root Existence and Monotonicity of $F_{\rm E}(\lambda)$}
We have
\begin{align*}
F_{\rm E}(0)&=\sum\limits_i q_i\log\dfrac{a_i}{q_i}-\zeta\Delta-R \\\
&< -\zeta\Delta+\sum\limits_i q_i\log\dfrac{a_i}{q_i}-R(\Delta,\bdq) \\
&\leq R(\Delta,\bdq)-R(\Delta,\bdq)=0.
\end{align*}
The first inequality holds because $R>R(\Delta,\bdq)$, otherwise by the definition of Marton's error exponent \eqref{marton-error-exponent}, $\bdq$ will be a feasible solution to problem \eqref{er-nested}, and the non-negativity of Kullback-Leibler divergence will imply that $\bdq$ is indeed the optimal solution. 
The second inequality follows directly from the dual definition of the RD function \eqref{dual-rate-distortion}. 
On the other hand, we have
\begin{align*}
F_{\rm E}'(\lambda)&=\dfrac{1}{(1+\lambda)^3}\dfrac{I_0I_2-I_1^2}{I_0^2}\geq 0, \\
I_k&=\sum\limits_i e^{\frac{\lambda\log a_i+\log q_i}{\lambda+1}}\left(\log\frac{a_i}{q_i}\right)^k,\quad k=0,1,2,
\end{align*}
where the inequality holds from Cauchy-Schwarz's inequality. 
Thus, we have shown that $F_{\rm E}(0)<0$ and $F_{\rm E}'(\lambda)\geq 0$, establishing the feasibility of Newton's method for root-finding. 

\subsubsection{The Monotonicity of $F_{\rm R}(\xi)$}
We have
\begin{align*}F_{\rm R}'(\xi)&=-\dfrac{1}{(1+\xi)^3}\dfrac{I_0I_2-I_1^2}{I_0^2}\leq 0, \\
I_k&=\sum\limits_i e^{\frac{\log a_i+\xi\log q_i}{1+\xi}}\left(\log\dfrac{a_i}{q_i}\right)^k,\quad k=0,1,2,
\end{align*}
where the inequality holds from Cauchy-Schwarz's inequality, hence $F_{\rm R}(\xi)$ is monotonic. 
Thus, we can employ Newton's method for root-finding in the case of $F_{\rm R}(0)>0$.


\section{Proof of Theorem 1 and 2}
We provide proofs of Theorem 1 and 2. 
Suppose the BA algorithm accurately solves the inherit RD problem and outputs a solution $(\bdw,\bdr)$, and define some auxiliary variables $u_{ij}^{(n)}=w_{ij}^{(n)}r_j^{(n)}$, $v_{ij}^{(n)}=u_{ij}^{(n)}/p_i^{(n)}$. 
Denote the $j$-th column of $\bdv$ as $\bdv_j$. 

\subsubsection{Proof of Theorem 1}
We have 
\[\begin{aligned}
a_i^{(n)}&=\dfrac{p_i^{(n)}}{\sum\limits_j e^{-\zeta d_{ij}}r_j^{(n)}}=\dfrac{p_i^{(n)}w_{ij}^{(n)}}{e^{-\zeta d_{ij}}r_j^{(n)}} =e^{\zeta d_{ij}}v_{ij}^{(n)},~\forall j.
\end{aligned}\]
So we have
{\fontsize{9pt}{12pt}\selectfont
\begin{align*}
&f_{\rm E}(\bdp^{(n+1)})-f_{\rm E}(\bdp^{(n)}) \\
=&\sum\limits_i p_i^{(n)}(\log p_i^{(n)}-\log q_i)-\sum\limits_i p_i^{(n+1)}(\log p_i^{(n+1)}-\log q_i) \\
=&\sum\limits_i p_i^{(n)}(\log p_i^{(n)}-\log q_i)-\sum\limits_i p_i^{(n+1)}(\log p_i^{(n+1)}-\log q_i) \\
&-\lambda^{(n+1)}\left(\zeta\Delta-\sum\limits_i p_i^{(n+1)}(\log a_i^{(n)}-\log p_i^{(n+1)})\right) \\
&+\lambda^{(n+1)}\left(\zeta\Delta-\sum\limits_i p_i^{(n)}(\log a_i^{(n-1)}-\log p_i^{(n)})\right) \\
=&\sum\limits_i p_i^{(n)}\left((\lambda^{(n+1)}+1)\log p_i^{(n)}-\lambda^{(n+1)}\log a_i^{(n-1)}-\log q_i\right) \\
-&\sum\limits_i p_i^{(n+1)}\Big((\lambda^{(n+1)}+1)\log p_i^{(n+1)}-\lambda^{(n+1)}\log a_i^{(n)}-\log q_i\Big) \\
=&\sum\limits_i p_i^{(n)}\left((\lambda^{(n+1)}+1)\log p_i^{(n)}-\lambda^{(n+1)}\log a_i^{(n)}-\log q_i\right) \\
&-(\lambda^{(n+1)}+1)\sum\limits_i p_i^{(n+1)}\log\left(\sum\limits_i e^{\frac{\lambda^{(n+1)}\log a_i^{(n)}+\log q_i}{\lambda^{(n+1)}+1}}\right) \\
&-\lambda^{(n+1)}\sum\limits_i p_i^{(n)}(\log a_i^{(n-1)}-\log a_i^{(n)}) \\
=&\sum\limits_i p_i^{(n)}\left((\lambda^{(n+1)}+1)\log p_i^{(n)}-\lambda^{(n+1)}\log a_i^{(n)}-\log q_i\right) \\
&-(\lambda^{(n+1)}+1)\sum\limits_i p_i^{(n)}\log\left(\sum\limits_ie^{\frac{\lambda^{(n+1)}\log a_i^{(n)}+\log q_i}{\lambda^{(n+1)}+1}}\right) \\
&-\lambda^{(n+1)}\sum\limits_i p_i^{(n)}(\log a_i^{(n-1)}-\log a_i^{(n)}) \\
=&(\lambda^{(n+1)}+1)\sum\limits_i p_i^{(n)}\log\frac{p_i^{(n)}}{p_i^{(n+1)}} \\
&\quad+\lambda^{(n+1)}\sum\limits_{i,j} p_i^{(n)}w_{ij}^{(n)}\log\frac{a_i^{(n)}}{a_i^{(n-1)}} \\
=&(\lambda^{(n+1)}+1)D(\bdp^{(n)}\Vert\bdp^{(n+1)})+\lambda^{(n+1)}\sum\limits_{i,j} r_j^{(n)}v_{ij}^{(n)}\log\dfrac{v_{ij}^{(n)}}{v_{ij}^{(n-1)}} \\
=&(\lambda^{(n+1)}+1)D(\bdp^{(n)}\Vert\bdp^{(n+1)})+\lambda^{(n+1)}\sum\limits_j r_j^{(n)}D(\bdv_j^{(n)}\Vert\bdv_j^{(n-1)}) \\
\geq& D(\bdp^{(n)}\Vert\bdp^{(n+1)}) \geq\dfrac{1}{2}\Vert\bdp^{(n)}-\bdp^{(n+1)}\Vert_1^2. 
\end{align*}
}
The last inequality is due to Pinsker's inequality. 
Since $\{f_{\rm E}(\bdp^{(n)})\}$ is increasing and upper bounded, we have
\[\dfrac{1}{2}\sum\limits_{n=1}^\infty\Vert\bdp^{(n)}-\bdp^{(n+1)}\Vert_1^2\leq\sum\limits_{n=1}^\infty\left(f(\bdp^{(n+1)})-f(\bdp^{(n)})\right)<+\infty,\]
so $\{\bdp^{(n)}\}$ converges to some limit point $\bdp^\star$. 
Next, noticing that $\{(\bdw^{(n)},\bdr^{(n)})\}$ is bounded, it has a subsequence $\{(\bdw^{(n_k)},\bdr^{(n_k)})\}$ converging to some $(\bdw^\star,\bdr^\star)$. 
The feasibility ensured by the BA algorithm implies that $r_j^\star=\sum\limits_i p_i^\star w_{ij}^\star$. 
The iterative scheme of $\bda$ implies that $\bda^{(n_k)}\to\bda^\star$ where $a_i^\star=p_i^\star\Bigg/\left(\sum\limits_j e^{-\zeta d_{ij}}r_j^\star\right)$. 
The convergence behavior together with the update scheme ensures that $(\bdp^\star,\bda^\star)$ is a feasible solution. 

The update scheme ensures that $(\bdp^{(n)},\bda^{(n)})$ satisfies the KKT condition in each iteration. 
The joint convexity of Kullback-Leibler divergence implies that problem \eqref{er-composite} is jointly convex with respect to $(\bdp,\bda)$, hence $(\bdp^\star,\bda^\star)$ is a global optimizer of \eqref{er-composite} when $\zeta$ is fixed.

\subsubsection{Proof of Theorem 2}
We have
{\fontsize{9pt}{12pt}\selectfont
\begin{align*}
&f_{\rm R}(\bdp^{(n+1)},\bda^{(n)})-f_{\rm R}(\bdp^{(n)},\bda^{(n)}) \\
=&\sum\limits_i p_i^{(n+1)}(\log a_i^{(n)}-\log p_i^{(n+1)})-\sum\limits_i p_i^{(n)}(\log a_i^{(n)}-\log p_i^{(n)}) \\
&-\xi^{(n+1)}\sum\limits_i p_i^{(n+1)}(\log p_i^{(n+1)}-\log q_i) \\
&+\xi^{(n+1)}\sum\limits_i p_i^{(n)}(\log p_i^{(n)}-\log q_i) \\
=&-\sum\limits_i p_i^{(n+1)}\left((\xi^{(n+1)}+1)\log p_i^{(n+1)}-\log a_i^{(n)}-\xi^{(n+1)}\log q_i\right) \\
&+\sum\limits_i p_i^{(n)}\left((\xi^{(n+1)}+1)\log p_i^{(n+1)}-\log a_i^{(n)}-\xi^{(n+1)}\log q_i\right) \\
&+(\xi^{(n+1)}+1)\sum\limits_i p_i^{(n)}(\log p_i^{(n)}-\log p_i^{(n+1)}) \\
=&-(\xi^{(n+1)}+1)\sum\limits_i (p_i^{(n+1)}-p_i^{(n)})\log\left(\sum\limits_i e^{\frac{\log a_i^{(n)}+\xi^{(n+1)}\log q_i}{1+\xi^{(n+1)}}}\right) \\
&+(\xi^{(n+1)}+1) D(\bdp^{(n)}\Vert\bdp^{(n+1)})=(\xi^{(n+1)}+1) D(\bdp^{(n)}\Vert\bdp^{(n+1)}). \\
&f_{\rm R}(\bdp^{(n+1)},\bda^{(n+1)})-f_{\rm R}(\bdp^{(n+1)},\bda^{(n)}) \\
=&\sum\limits_i p_i^{(n+1)}(\log a_i^{(n+1)}-\log a_i^{(n)}) \\
=&\sum\limits_{i,j} p_i^{(n+1)}w_{ij}^{(n+1)}\log\dfrac{v_{ij}^{(n+1)}}{v_{ij}^{(n)}}=\sum\limits_j r_j^{(n+1)}D(\bdv_j^{(n+1)}\Vert\bdv_j^{(n)}).
\end{align*}
}
In conclusion, we have
\[\begin{aligned}
&f_{\rm R}(\bdp^{(n+1)},\bda^{(n+1)})-f_{\rm R}(\bdp^{(n)},\bda^{(n)}) \\
=&(\xi^{(n+1)}+1)D(\bdp^{(n)}\Vert\bdp^{(n+1)})+\sum\limits_j r_j^{(n+1)}D(\bdv_j^{(n+1)}\Vert\bdv_j^{(n)}) \\
\geq&D(\bdp^{(n)}\Vert\bdp^{(n+1)})\geq\dfrac{1}{2}\Vert\bdp^{(n)}-\bdp^{(n+1)}\Vert_1^2.
\end{aligned}\]
Similar to the above analysis, since $\{f_{\rm R}(\bdp^{(n)},\bda^{(n)})\}$ is increasing and upper bounded, the sequence $\{\bdp^{(n)}\}$ converges to some limit point $\bdp^\star$. 
A parallel argument to that developed the preceding subsection shows that $(\bdp^\star,\bda^\star)$ globally optimizes problem \eqref{re-composite} with fixed $\zeta$. 
%

%
%

\end{appendices}

\end{document}